\documentclass[html]{article}
\usepackage{amsmath}
\usepackage{graphicx}
\usepackage{algorithmic}
\usepackage{algorithm}
\usepackage{hyperref}
\usepackage{verbatim}
\usepackage{enumerate}
\usepackage{amsfonts}
\usepackage{amsthm, amssymb}
\usepackage{epstopdf}
\usepackage{subfig}
\usepackage{color}

\newtheorem{theorem}{Theorem}

\title{On Optimal Disc Covers and \\a New Characterization of the Steiner Center}

\usepackage{authblk}
\author{Yael Yankelevsky and Alfred M. Bruckstein}
\affil{Technion - Israel Institute of Technology, Haifa 32000, Israel}
\date{}

\begin{document} 
\maketitle
\begin{abstract}
Given $N$ points in the plane $P_1,P_2,...,P_N$ and a location $\Omega$, the union of discs with diameters $\left[\Omega P_i\right]$, $i=1,2,...,N$ covers the convex hull of the points.
The location $\Omega_s$ minimizing the area covered by the union of discs, is shown to be the Steiner center of the convex hull of the points. Similar results for $d$-dimensional Euclidean space are conjectured.
\end{abstract}

\section{Introduction}
In this paper we discuss a sphere coverage problem and, in this context, we propose an optimal coverage criterion defining a center for a given set of points in space.

Suppose that a constellation of $N$ points $\left\{P_1,P_2,...,P_N\right\}$ in $\mathbb{R}^d$ (the $d$-dimensional Euclidean space) is given. 
An arbitrary point $\Omega\in \mathbb{R}^d$ is selected and the spheres $S_{P_i}(\Omega)$, having [$\Omega P_i$] as diameters, are defined.
Hence the centers of $S_{P_i}(\Omega)$ are at 
$\frac{1}{2}(\Omega+P_i)$ and their radii are $\frac{1}{2}\|\Omega-P_i\|$, $i=1,2,...,N$. 

Consider the union of these spheres $S_{P_i}(\Omega)$, their surface "anchored" at $\Omega$.
First we prove that the resulting $d$-dimensional shape always covers the convex hull $CH\left\{P_1,P_2,...,P_N\right\}$ of the given points, hence its volume exceeds the volume of this convex hull for all $\Omega\in\mathbb{R}^d$.
This leads to the following natural question: what is the location $\Omega^*$ which minimizes the excess (or overflow) volume and hence the total volume of the shape, $\Sigma_{(\Omega)}=\bigcup_{i=1}^{N}{S_{P_i}(\Omega)}$?

Such a location, we claim, would be a natural candidate as a "center" for the constellation of points $\left\{P_1,P_2,...,P_N\right\}$.

The problem of determining the point that gives the tightest cover with spheres, minimizing the excess volume beyond the convex hull, is solved here for the planar case (i.e. $d=2$).
An illustration of this problem is presented in Figure~\ref{Fig:problemStatement}.
\begin{figure}[ht]
\centering \includegraphics[scale=0.45]{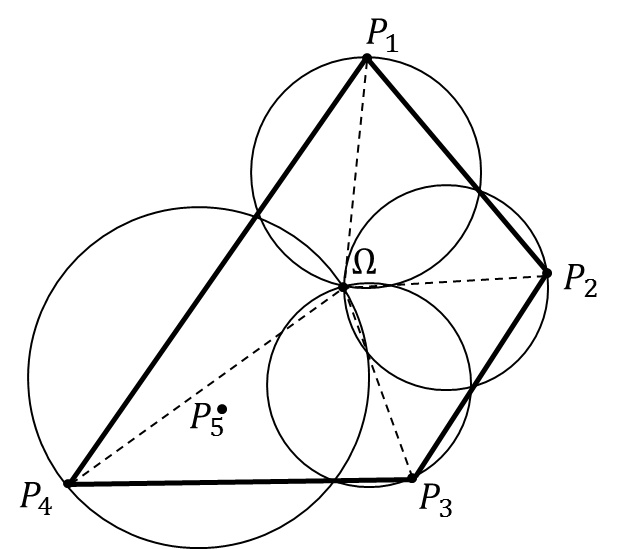} 
\caption{Illustration of the problem in 2 dimensions: For the set of points $P_1,...,P_5$, the (highlighted) convex hull polygon is defined by the vertices  $P_1,...,P_4$. These vertices define 4 discs anchored at the arbitrary point $\Omega$}
\label{Fig:problemStatement}
\end{figure}
The result is the following:
the optimal location $\Omega^*$, is the so called Steiner center of the convex hull of the given points $\left\{P_1,P_2,...,P_N\right\}\in\mathbb{R}^2$. 
The Steiner center is a weighted centroid of the vertices of a convex polygon, the weights being proportional to the exterior angles at the vertices (see Figure~\ref{Fig:extAngles}).
Hence, the Steiner center $\Omega_s$ of a convex polygon $[V_1 V_2 ... V_k]$ is also characterized as the point that yields the tightest disc cover with discs having $[\Omega_sV_j]$ as diameters ($j=1,2,...k$).

For the $d$-dimensional case we conjecture that a similar result holds, however a proof is yet to be found. Some numerical simulations that were performed in 3D seem to confirm this conjecture.

\subsection{Centers for Point Constellations}
Finding meaningful centers for a collection of data points is a fundamental geometric problem in various data analysis and operation research/facility location applications.

One of the interesting centers is the Steiner point (also known as the Steiner curvature centroid).
The Steiner point  of a convex polygon in $\mathbb{R}^2$, is defined as the weighted centroid (i.e. center of mass) of the system obtained by placing a mass equal to the magnitude of the exterior angle at each vertex \cite{Honsberger}.
The traditional characterization is therefore
\begin{equation} \label{Eq:SteinerTraditional}
\Omega_s = \arg\min_{\Omega}\sum_{i=1}^{k}{\theta_i d^2(V_i,\Omega)}
\end{equation}
yielding explicitly 
\begin{equation} 
\Omega_s = \frac{1}{2\pi}\sum_{i=1}^{k}{\theta_i V_i}
\end{equation}
where $d(V_i,\Omega)$ is the Euclidean distance from $V_i$ to $\Omega$ and $\theta_i$ are the external turn angles at the vertices $V_i$ of the convex polygon, that sum to $2\pi$ (see Figure~\ref{Fig:extAngles}).

\begin{figure}[ht]
\centering \includegraphics[scale=0.5]{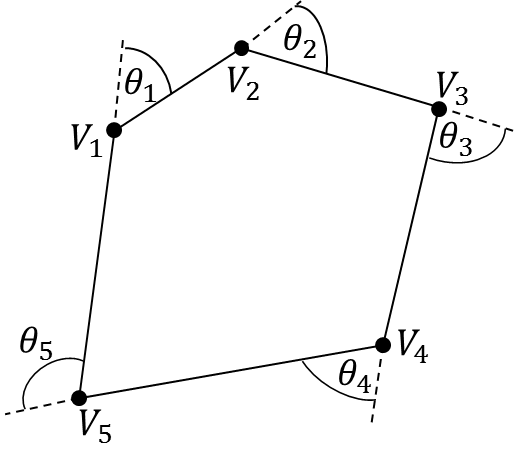} 
\caption{External turn angles}
\label{Fig:extAngles}
\end{figure}

Another characterization of the Steiner center is by projections \cite{Durocher}.
Let $P_i^\theta$ denote the projection of the point $P_i$ on the unit vector $u_\theta=(\cos\theta,\sin\theta)$:
\begin{equation}
	P_i^\theta=u_\theta<P_i,u_\theta>
\end{equation}
then the Steiner center is defined as:
\begin{equation}
	\Omega_s = \frac{1}{\pi}\int_0^{\pi}{u_\theta \left(\operatorname*{\,min}_{i}{|P_i^\theta|}+\operatorname*{\,max}_{i}{|P_i^\theta|}\right)d\theta}
\end{equation}

Furthermore, the Steiner center $\Omega_s$ of a convex shape has some very interesting properties, the nicest one being its linearity with respect to Minkowski addition.
Hence, if $K_1$ and $K_2$ are two convex sets in $\mathbb{R}^d$, we have that
\begin{equation}
\Omega_s(K_1\oplus K_2)=\Omega_s(K_1)+\Omega_s(K_2)
\end{equation}
where $\oplus$ stands for vector addition, i.e.
\begin{equation}
K_1\oplus K_2 = \left\{x+y|x\in K_1,y\in K_2\right\}  .
\end{equation}
It is also true that the map $K\rightarrow\Omega_s(K)$ is similarity invariant, i.e.
\begin{equation}
\Omega_s(tK)=t\Omega_s(K)
\end{equation}
where 
\begin{equation}
tK = \left\{tx|x\in K\subset\mathbb{R}^d\right\} , t>0, 
\end{equation}
and it is well known (see Shephard \cite{Shephard1,Shephard2}, Sallee \cite{Sallee} and Schneider \cite{Schneider}) that these properties and continuity of the mapping characterize the Steiner point.
\newline

The Steiner center, along with other suggested centers for point constellations  (such as the center of gravity, the centroid of the convex hull and the Weber-Fermat median), were all subject to intense research see e.g. \cite{Alt1994},\cite{Berg},\cite{Carmi},\cite{Durocher}, \cite{Kaiser},\cite{McMullen},\cite{Shephard3}).
All these points are characterized by various optimization criteria, such as (weighted) sums of distances (or functions of distances) to the given points or  minimax criteria with different metrics. 

However, we have never encountered a "center" location optimization criterion expressed as the area of a union of shapes defined in terms of the variable point $\Omega$ and the points of the given data set. We note that the problem of covering the convex hull of a set of points with unions of spheres, $\Sigma_{(\Omega)}=\bigcup_{i=1}^{N}{S_{P_i}(\Omega)}$, arose in the analysis of monitoring threshold functions over distributed data streams, in the work of Sharfman, Schuster and Keren \cite{Keren2012,Keren2006}. 
In this work, the authors provided a proof of the coverage result based on a variant of Carath\'{e}odory's theorem, using induction on the dimensionality $d$. The proof we present here is simple and direct, and does not rely on any results beyond the definition of convexity.

After this paper was submitted we found out that in a seminal work on the complexity of computing the volume, Elekes \cite{Elekes1986} considered the same issue and provided a simple proof of coverage very similar to the one we present below (we thank Prof. J. Pach for pointing out Elekes'  paper to us, following a presentation of this work).
\newline

The rest of the paper is organized as follows:
Section~\ref{Section:RdCoverage} proves the theorem on coverage of the convex hull in $\mathbb{R}^d$, then Section~\ref{Section:R2} analyzes the problem for the plane ($d=2$) and presents an even simpler argument proving convex hull coverage and shows that the optimal $\Omega$ is the Steiner point of the convex hull of a planar constellation of points.
Finally, Section~\ref{Section:Conclusions} offers some concluding remarks.

\section{$d$-dimensional sphere covers} \label{Section:RdCoverage}
Given a set of points in $\mathbb{R}^d$, denoted by $\{P_1,P_2,...,P_N\}$, for any 
$\Omega\in\mathbb{R}^d$ define the spheres $S_{P_i}(\Omega)$ with center at the midpoint of the segment $\left[\Omega P_i\right]$ and radius $\frac{1}{2}\|\Omega P_i\|$.
We prove the following:

\begin{theorem} \label{Theorem:Th1}
(Elekes, \cite{Elekes1986})
\begin{equation}
CH\left\{P_1,P_2,...,P_N\right\} \subset \bigcup_{i=1}^{N}{S_{P_i}(\Omega)}
\end{equation}
where $CH\left\{P_1,P_2,...,P_N\right\}$ denotes the convex hull of the points $P_1,P_2,...,P_N$.
\end{theorem}

Without loss of generality, we choose the coordinate system such that $\Omega$ is the origin,  i.e. $\Omega=(0,0,...0)\in\mathbb{R}^d$.
Denote a general point in the convex hull of $\{P_1,P_2,...,P_N\}$ by $Q=\sum_{i=1}^{N}{\lambda_iP_i}$ (with $\lambda_i\geq0,\sum_{i=1}^{N}{\lambda_i}=1$) .

To prove the inclusion of the convex hull in the union of the spheres $S_{P_i}(\Omega)$ we must show that:
\begin{equation} 
	\exists{i\in\{1,2,...,N\}} \quad s.t. \quad d(Q,\frac{1}{2}P_i)\leq d(\Omega,\frac{1}{2}P_i)
\end{equation}

hence Q is inside at least one of the spheres, being closer to the sphere center than its radius.
This clearly implies that:
\begin{equation} 
    Q \in CH\left\{P_1,P_2,...,P_N\right\} \Rightarrow Q\in\bigcup_{i=1}^{N}{S_{P_i}(\Omega)}
\end{equation}

\begin{proof}[Proof of Theorem~\ref{Theorem:Th1}]
Assume that
\begin{equation} \label{Eq:assumption} 
	d(Q,\frac{1}{2}P_i) > d(\Omega,\frac{1}{2}P_i)  \quad \forall i\in\{1,2,...,N\}
\end{equation}

Hence we have
\begin{align} 
	& d^2(Q,\frac{1}{2}P_i) > d^2(\Omega,\frac{1}{2}P_i) \notag \\
	& (Q-\frac{1}{2}P_i)^T(Q-\frac{1}{2}P_i) > \frac{1}{2}P_i^T\cdot \frac{1}{2}P_i \notag \\
	& Q^T Q-Q^TP_i > 0 \notag \\
	& Q^T(P_i-Q) < 0 \quad \forall i\in\{1,2,...,N\}
\end{align} 

This means that the projections of all the vectors from $Q$ to $P_i$ $(=\vec{P_i}-\vec{Q})$, on the vector from $\Omega$ to $Q$ ($=\vec{ Q}$) are strictly negative (see Figure~\ref{Fig:RdProof}).
But this is impossible since $Q\in CH\{P_1,P_2,...,P_N\}$ and this implies that $CH\{P_1,P_2,...,P_N\}$ cannot project on the line $\Omega Q$ on "one side" of $Q$.

\begin{figure}[H]
\centering\includegraphics[scale=0.5,clip, trim= 0cm 1.5cm 0cm 0cm]{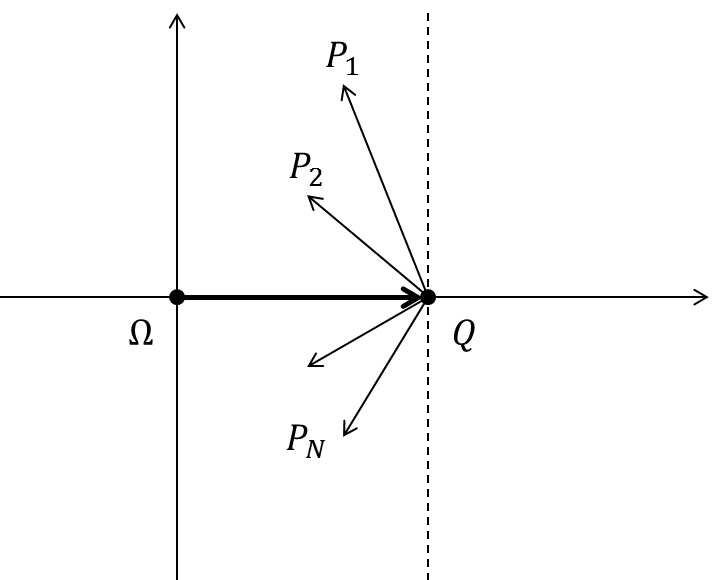}
\caption{Strictly negative projections} \label{Fig:RdProof}
\end{figure}

The contradiction to the assumption in \eqref{Eq:assumption} proves that we must have for some $i$:
\begin{equation} 
	d(Q,\frac{1}{2}P_i) \leq d(\Omega,\frac{1}{2}P_i)
\end{equation}
Hence
\begin{equation} 
    Q \in CH\left\{P_1,P_2,...,P_N\right\} \Rightarrow Q\in\bigcup_{i=1}^{N}{S_{P_i}(\Omega)} 
\end{equation}
\end{proof}

Since $CH\left\{P_1,P_2,...,P_N\right\}\subset\bigcup_{i=1}^{N}{S_{P_i}(\Omega)}$ , we have that
\begin{equation} \nonumber
\begin{aligned}
Volume &\left(\bigcup_{i=1}^{N}{S_{P_i}(\Omega)}\right) = Volume \left(CH\left\{P_1,...,P_N\right\}\right) \\
+& Volume \left( \bigcup_{i=1}^{N}{S_{P_i}(\Omega)} \setminus CH\left\{P_1,P_2,...,P_N\right\}   \right)
\end{aligned}
\end{equation}

It therefore makes sense to ask what is the location $\Omega^*$ that minimizes the volume of the union of spheres  $S_{P_i}(\Omega)$, hence also the excess volume beyond the convex hull of the data points.
In the next section we solve this problem for the important planar case ($d=2$).
Surprisingly, the optimal location $\Omega$ turns out to be a well-known center for planar convex shapes, the Steiner center.

\section{A discovery on disc covers} \label{Section:R2}
In this section, we analyze the planar disc covering problem, first providing an even simpler proof of the convex hull coverage result (Theorem~\ref{Theorem:Th1}) and then determining the location of $\Omega$ that results in the tightest cover.
Namely, we show the following:
\emph{
Given $\left\{V_1,V_2,...,V_k\right\}$ the vertices of a convex polygon in $\mathbb{R}^2$, the Steiner point $\Omega_s$ is the solution of
}
\begin{equation} 
\Omega_s = \arg\min_{\Omega}\left\{Area\left( \bigcup_{i=1}^{k}{S_{V_i}(\Omega)} \right) \right\}
\end{equation}

\subsection{$CH\left\{P_1,P_2,...,P_N\right\}$ is covered by the union of discs $\bigcup_i{S_{P_i(\Omega)}}$}
In 2D, each pair of discs $i,j\in\{1,...,N\}$ may have one of the following mutual positions:
\begin{enumerate}
\item{The boundary circles are tangent to each other at the point $\Omega$.

It is readily seen from Figure~\ref{Fig:tangentCircles} that in this case, the segment $[P_iP_j]$ is either entirely included in a single disc, or the common tangent line through $\Omega$ is perpendicular to both diameters and so $[\Omega P_i],[\Omega P_j]$ are collinear, such that $[P_iP_j]$ consists of the 2 diameters and hence belongs to the union of the 2 discs.

\begin{figure}[ht]
\centering
\subfloat[]{
	\centering
    \includegraphics[scale=0.5]{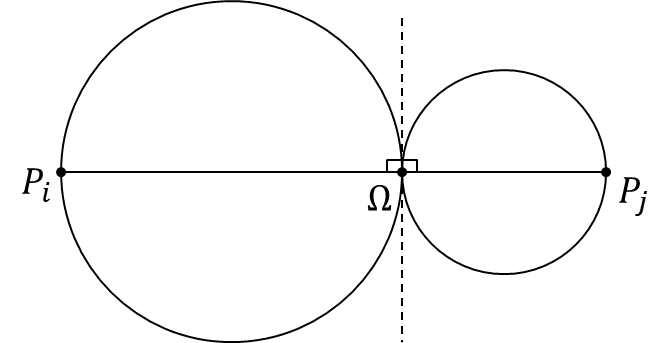}
}
\subfloat[]{
	\centering
    \includegraphics[scale=0.5]{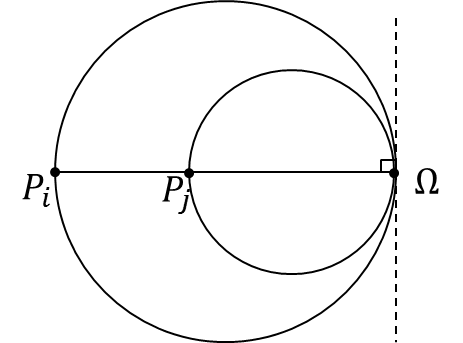}
}
\caption[Tangent circles positions] {Tangent circles (case 1)}
\label{Fig:tangentCircles}
\end{figure}

} 

\item{The circles intersect at two points: $\Omega$ and $Q$ ($Q\neq\Omega$)

Since every inscribed angle that subtends a diameter is a right angle, we have $\angle\Omega QP_i=\angle\Omega QP_j=\frac{\pi}{2}$. Hence either $Q\in[P_iP_j]$ or $Q$ is outside the segment $[P_iP_j]$ but on the same line.
We clearly see that in both cases the segment $[P_iP_j]$ and the triangle  $\bigtriangleup \Omega P_i P_j$ are covered by the union of the 2 discs (see Figure~\ref{Fig:intersectingCircles}).

\begin{figure}[ht]
\centering
\subfloat[]{
	\centering
    \includegraphics[scale=0.55]{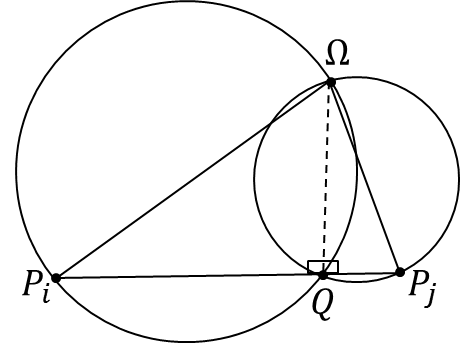}
}
\subfloat[]{
	\centering
    \includegraphics[scale=0.55]{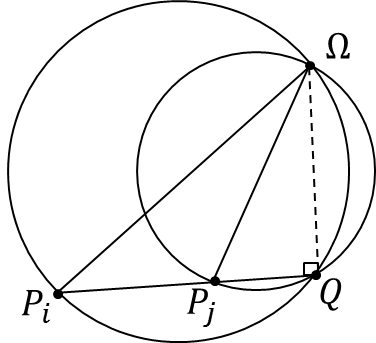}
}
\caption[Intersecting circles positions] {Intersecting circles (case 2)}
\label{Fig:intersectingCircles}
\end{figure}

} 
\end{enumerate}

So far it was shown that for every pair of discs $i,j$, the line segment $[P_iP_j]$, and in fact the triangle $\bigtriangleup \Omega P_i P_j$, is covered by the union of the 2 discs.

The convex hull of a finite set of points in $\mathbb{R}^2$ is a convex polygon whose vertices are a subset of the point set $\left\{P_1,P_2,...,P_N\right\}$. Therefore the CH polygon edges are a subset of all possible segments $\{[P_iP_j] \quad \forall i,j\}$.
As each such segment, and hence each polygon edge, belongs to the union of 2 discs, it obviously belongs to the union of all discs.

Since all the discs intersect at $\Omega$, the union of discs is a star-shaped region, i.e. 
\begin{equation}
\forall Q_0\subset\bigcup_{i=1}^{N}{S_{P_i}(\Omega)},\quad [\Omega Q_0]\subset\bigcup_{i=1}^{N}{S_{P_i}(\Omega)}
\end{equation}
Due to this fact, together with the convexity of the CH polygon, the CH  is completely covered by the union of triangles $\bigcup_{i,j=1}^{N}{\bigtriangleup \Omega P_i P_j}$. 
Finally, since each such triangle is covered by the union of discs, it follows that
\begin{equation}
    \forall \Omega : \quad CH\{P_1,...,P_N\} \subset \bigcup_{i=1}^{N}{S_{P_i}(\Omega)}
\end{equation}
\qed

\subsection{The optimal location for $\Omega$} \label{Section:OptimalOmega}
Next, let us determine the optimal location of $\Omega$ in the sense of minimizing the area difference between the union of discs $S_{P_i}(\Omega),\;i=1,2,...,N$ and the convex hull $CH\left\{P_1,P_2,...,P_N\right\}$. 
Clearly this requires us to simply minimize the area of $\bigcup_{i=1}^{N}{S_{P_i}(\Omega)}$.

Denote by $\Delta S(\Omega)$ the "overflow" region covered beyond $CH\left\{P_1,P_2,...,P_N\right\}$, i.e.
\begin{equation} \label{Eq:DS}
\Delta S(\Omega)=\bigcup_{i=1}^{N}{S_{P_i}(\Omega)}\setminus CH\left\{P_1,P_2,...,P_N\right\}
\end{equation}

\begin{theorem} \label{Theorem:Th2}
The area of $\Delta S(\Omega)$ is minimized when $\Omega$ is located at the Steiner center of the convex hull of $\left\{P_1,P_2,...,P_N\right\}\in\mathbb{R}^2$.
\end{theorem}

\begin{proof}[Proof of Theorem~\ref{Theorem:Th2}]\hspace*{\fill}

We first consider $\Omega\in CH\left\{P_1,P_2,...,P_N\right\}$, which after reordering and renumbering the extremal points from $\left\{P_1,P_2,...,P_N\right\}$ is a convex polygon defined by $\left\{\bar{P}_1,\bar{P}_2,...,\bar{P}_M\right\}:\quad (\bar{P}_1\rightarrow \bar{P}_2\rightarrow ...\rightarrow \bar{P}_M\rightarrow \bar{P}_1)$.

It is readily seen that the $N-M$ points in the interior of the convex hull polygon define discs that are covered by the M discs determined by the external points.
Indeed, if $P_k$ is a point in $\left\{P_1,P_2,...,P_N\right\} \setminus \left\{\bar{P}_1,\bar{P}_2,...,\bar{P}_M\right\}$ we have that $S_{P_k}(\Omega) \subset S_{\tilde{P}_k}(\Omega)$ where $\tilde{P}_k$ is the point where the ray $\left[\Omega P_k\right)$ exits the convex hull (see Figure~\ref{Fig:PK}).
\begin{figure}[ht]
\centering \includegraphics[scale=0.47]{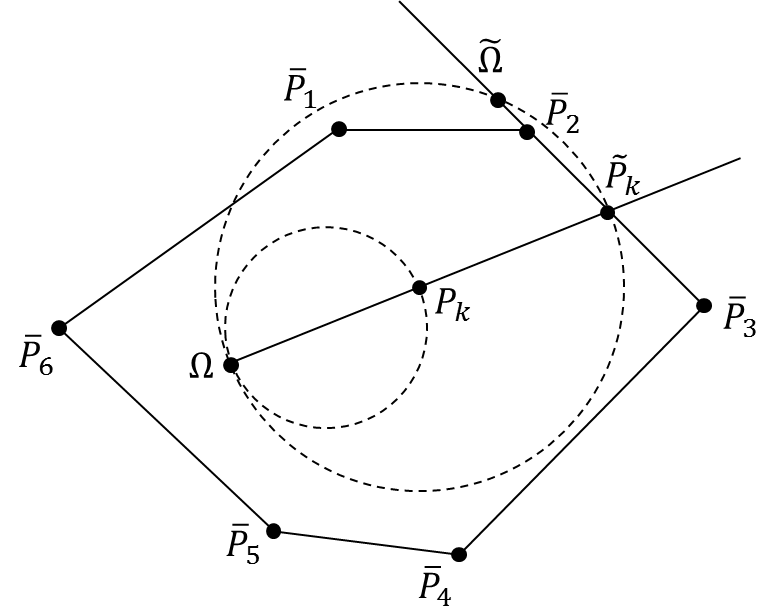}
\caption{Discs defined by internal points are covered by discs defined by external points} \label{Fig:PK}
\end{figure}

The point $\tilde{P}_k$ is on a boundary segment $\left[\bar{P}_\ell \bar{P}_{\ell+1}\right]$ of the convex hull and $S_{\bar{P}_\ell}(\Omega)\cup S_{\bar{P}_{\ell+1}}(\Omega)$ clearly covers $S_{\tilde{P}_k}(\Omega)$, since all three circles intersect at $\Omega$ and at its projection on the line $\left(\bar{P}_\ell \bar{P}_{\ell+1}\right)$, denoted by $Q_\ell$ (see Figure~\ref{Fig:Ql}).
\begin{figure}[ht]
\centering \includegraphics[scale=0.47]{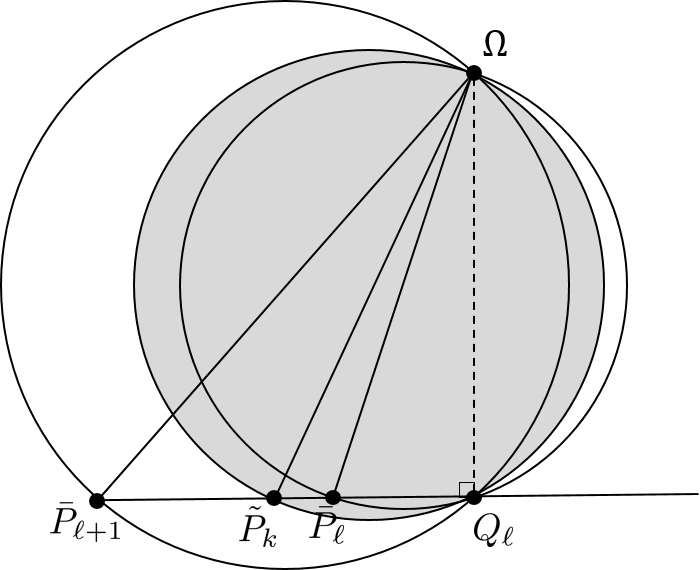}
\caption{The discs $S_{\bar{P}_\ell}(\Omega)$,$S_{\bar{P}_{\ell+1}}(\Omega)$ and $S_{\tilde{P}_k}(\Omega)$ intersect at $\Omega$ and $Q_\ell$} \label{Fig:Ql}
\end{figure}

Therefore let us define the shape $S:=\bigcup_{i=1}^{M}S_{\bar{P}_i}(\Omega)$ and compute its area explicitly as a function of the location of $\Omega$.

Consider the convex polygon $\bar{P}_1 \rightarrow\bar{P}_2 \rightarrow ... \rightarrow\bar{P}_M\rightarrow\bar{P}_1$ and the point $\Omega$ inside it (see Figure~\ref{Fig:OmegaInPoly}).

\begin{figure}[h]
\centering \includegraphics[scale=0.45,clip,trim=0cm 0.1cm 0cm 0.2cm]{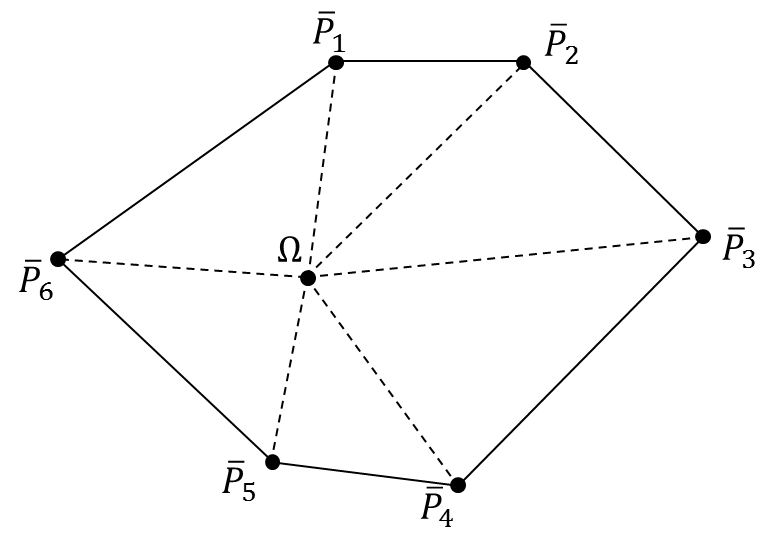}
\caption{The convex polygon $\bar{P}_1\rightarrow\bar{P}_2\rightarrow ...\rightarrow\bar{P}_M\rightarrow\bar{P}_1$ (for $M=6$) with an internal point $\Omega$} \label{Fig:OmegaInPoly}
\end{figure}

The diameters $\left[\Omega\bar{P}_i\right]$ are segments that form a "star configuration" about $\Omega$, their length being $d_i:= d\left[\Omega\bar{P}_i\right]$.
Let us denote by $Q_i$ the projections of $\Omega$ on the lines $\left(\bar{P}_i\bar{P}_{i+1}\right)$. For that purpose, we set $\bar{P}_{M+1}:=\bar{P}_1$.
Also define the angles
\begin{equation} \nonumber
\begin{cases}
\angle\bar{P}_i\Omega Q_i &= \alpha_i \\
\angle Q_i\Omega \bar{P}_{i+1} &= \beta_{i+1}
\end{cases}
\quad i=1,2,...,M
\end{equation}
as illustrated in Figure~\ref{Fig:defineAngles}.
\newline

\begin{figure}[h]
\centering
\subfloat[]{
	\centering
    \includegraphics[scale=0.6]{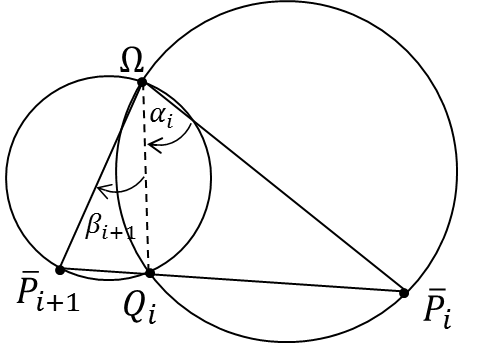}
\label{Fig:defineAngles1}
}
\subfloat[]{
	\centering
    \includegraphics[scale=0.6]{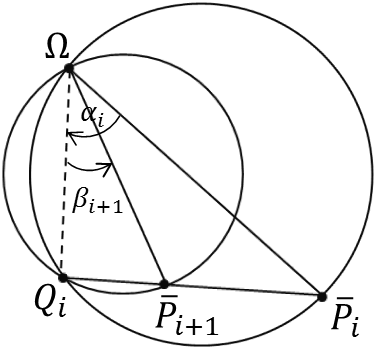}
\label{Fig:defineAngles2}
}
\subfloat[]{
	\centering
    \includegraphics[scale=0.6]{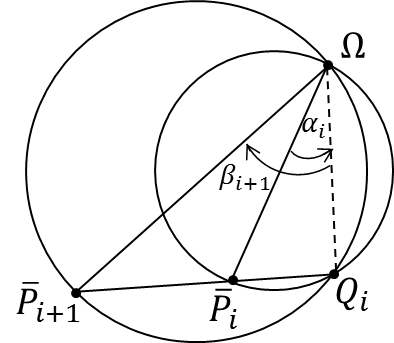}
\label{Fig:defineAngles3}
}
\caption[Internal angles definition] {Definition of the angles $\alpha_i$, $\beta_{i+1}$ for the different possible locations of $Q_i$}
\label{Fig:defineAngles}
\end{figure}

Note that $\alpha_i$ is defined \emph{towards} $Q_i$ and $\beta_i$ \emph{from} $Q_i$, so their directions (clockwise or counter-clockwise) may be inconsistent and depend on the geometric configuration (different configurations can be seen in Figure~\ref{Fig:defineAngles}).

We recall (see Figure~\ref{Fig:simpleCircle}) that the area of a circular segment is given by 
\begin{equation} \label{Eq:simpleCircSegment}
S_{(segment)}(QP)=\frac{1}{4}d^2\alpha-S(\triangle POQ)=\frac{1}{4}d^2\alpha-\frac{1}{2}S(\triangle P\Omega Q).
\end{equation}

\begin{figure}[ht]
\centering \includegraphics[scale=0.53, clip, trim= 0cm 1.7cm 0cm 1cm]{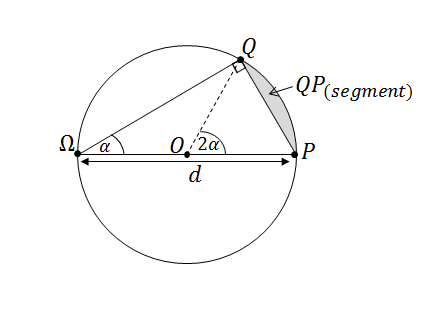}
\caption{Basic properties of triangles and circular segments} \label{Fig:simpleCircle}
\end{figure}

With these preliminary definitions and basic facts in mind, we can calculate the area of the union of discs $\bigcup_{i=1}^{M}S_{\bar{P}_i}(\Omega)$ and the area of the convex hull $CH\left\{\bar{P}_1,\bar{P}_2,...,\bar{P}_M\right\}$ in terms of the distances $d_i$ and the angles $\alpha_i$ and $\beta_i$ (see Figure~\ref{Fig:complexAngles}).
\newline
\begin{figure}[ht]
\centering \includegraphics[scale=0.6,trim=0cm 0cm 0cm 0.3cm,clip]{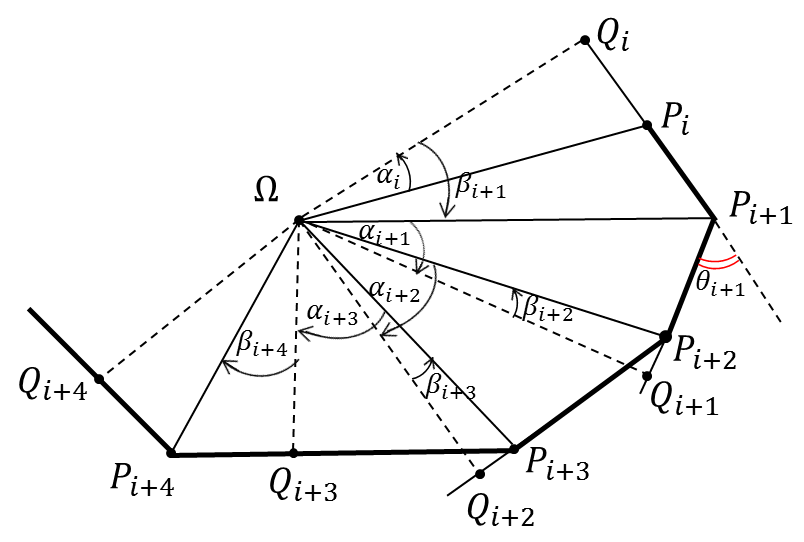}
\caption{Computing the excess area $\Delta S$ and the area of the convex hull $S_{CH}$ (note that $\theta_{i+1}=\angle Q_i\Omega Q_{i+1}=\angle Q_i\Omega P_{i+1}+\angle P_{i+1}\Omega Q_{i+1} = \beta_{i+1}+\alpha_{i+1}$)} \label{Fig:complexAngles}
\end{figure}

Let us express the excess area $\Delta S$ defined in \eqref{Eq:DS} as a sum of circular segments.
It can be seen from Figure~\ref{Fig:defineAngles} that the excess area over the CH edge $[\bar{P}_i\bar{P}_{i+1}]$ in the three possible scenarios is either the sum or difference of the circular segments lying on the chords $[Q_i\bar{P}_i],[Q_i\bar{P}_{i+1}]$.

If $Q\in[\bar{P}_i\bar{P}_{i+1}]$ (see Figure~\ref{Fig:defineAngles1}), the excess area over $[\bar{P}_i\bar{P}_{i+1}]$, denoted $\Delta S_i$, is 
\begin{equation} \nonumber
\Delta S_i = S_{(segment)}\left(Q_i\bar{P}_i\right)+S_{(segment)}\left(Q_i\bar{P}_{i+1}\right)
\end{equation}
and 
\begin{equation} \nonumber
	S(\triangle \bar{P}_i\Omega \bar{P}_{i+1})=S(\triangle \bar{P}_i\Omega Q_i)+ S(\triangle Q_i\Omega \bar{P}_{i+1})
\end{equation}
Thus using \eqref{Eq:simpleCircSegment},
\begin{equation}
\begin{aligned}
\Delta S_i &= \left[\frac{1}{4}d_i^2\alpha_i-\frac{1}{2}S(\triangle \bar{P}_i\Omega Q_i)\right] +\left[\frac{1}{4}d_{i+1}^2\beta_{i+1}-\frac{1}{2}S(\triangle Q_i\Omega\bar{P}_{i+1})\right] \\
=& \frac{1}{4}d_i^2\alpha_i+\frac{1}{4}d_{i+1}^2\beta_{i+1}-\frac{1}{2}\left[S(\triangle \bar{P}_i\Omega Q_i)+S(\triangle Q_i\Omega\bar{P}_{i+1})\right] \\
=& \frac{1}{4}d_i^2\alpha_i+\frac{1}{4}d_{i+1}^2\beta_{i+1}-\frac{1}{2}S(\triangle \bar{P}_i\Omega \bar{P}_{i+1})
\end{aligned}
\end{equation}
If $Q\not\in[\bar{P}_i\bar{P}_{i+1}]$ and is on the continuation of the line determined by $[\bar{P}_i\bar{P}_{i+1}]$ beyond $\bar{P}_{i+1}$ (see Figure~\ref{Fig:defineAngles2}),  
\begin{equation} \nonumber
\Delta S_i = S_{(segment)}\left(Q_i\bar{P}_i\right)-S_{(segment)}\left(Q_i\bar{P}_{i+1}\right)
\end{equation}
and 
\begin{equation} \nonumber
	S(\triangle \bar{P}_i\Omega \bar{P}_{i+1})=S(\triangle \bar{P}_i\Omega Q_i)- S(\triangle Q_i\Omega \bar{P}_{i+1})
\end{equation}
resulting in  
\begin{equation}
\begin{aligned}
\Delta S_i &= \left[\frac{1}{4}d_i^2\alpha_i-\frac{1}{2}S(\triangle \bar{P}_i\Omega Q_i)\right] -\left[\frac{1}{4}d_{i+1}^2\beta_{i+1}-\frac{1}{2}S(\triangle Q_i\Omega\bar{P}_{i+1})\right] \\
=& \frac{1}{4}d_i^2\alpha_i-\frac{1}{4}d_{i+1}^2\beta_{i+1}-\frac{1}{2}S(\triangle \bar{P}_i\Omega \bar{P}_{i+1})
\end{aligned}
\end{equation}
From symmetry, if $Q_i$ is on the line determined by $[\bar{P}_{i+1}\bar{P}_{i}]$ beyond $\bar{P}_{i}$ (see Figure~\ref{Fig:defineAngles3}), the result is identical (up to a sign).

Therefore, summing the excess area over all the M discs, we can write:
\begin{equation}\label{Eq:totalArea}
\begin{aligned}
	\Delta S &= \sum_{i=1}^{M}{ \Delta S_i } = \sum_{i=1}^{M}{ \frac{d_i^2}{4}\alpha_i}+\sum_{i=1}^{M}{\pm\frac{d_{i+1}^2}{4}\beta_{i+1}} -\frac{1}{2}\sum_{i=1}^{M}{S(\triangle \bar{P}_i\Omega \bar{P}_{i+1}) }\\
	&= \sum_{i=1}^{M}{ \Delta S_i } = \sum_{i=1}^{M}{ \frac{d_i^2}{4}\left(\alpha_i\pm\beta_i\right)} -\frac{1}{2}\sum_{i=1}^{M}{S(\triangle \bar{P}_i\Omega \bar{P}_{i+1}) }
\end{aligned}
\end{equation} 
We observe that $\angle Q_{i-1}\Omega Q_i = \alpha_i\pm\beta_i$ (with the sign $\pm$ depending on $i$), and that the convex hull area can similarly be expressed as a sum of triangles:
\begin{equation}\label{Eq:CHarea}
	S_{CH} = \sum_{i=1}^{M}{S(\triangle \bar{P}_i\Omega \bar{P}_{i+1})} 
\end{equation}
Given those observations, we can rewrite \eqref{Eq:totalArea} as:
\begin{equation} 
	\Delta S = \sum_{i=1}^{M}{ \frac{d_i^2}{4}\left(\angle Q_{i-1}\Omega Q_i\right)} -\frac{1}{2}S_{CH}
\end{equation}

Now we wish to find the optimal center that yields the minimal excess area, and since $S_{CH}$ is independent of $\Omega$, we need to solve:
\begin{equation} \label{Eq:simplerProblem}
	\Omega^* = \arg\min_\Omega \Delta S = \arg\min_\Omega \sum_{i=1}^{M}{d_i^2 \left(\angle Q_{i-1}\Omega Q_i\right)}
\end{equation}
By extending each edge $[\bar{P}_{i-1} \bar{P}_i]$ outside the polygon, an exterior angle is formed at the vertex $\bar{P}_i$ whose size is exactly 
$\theta_i=\angle Q_{i-1}\Omega Q_i$ (see Figure~\ref{Fig:complexAngles}). As $\theta_i$ are independent of $\Omega$, \eqref{Eq:simplerProblem} becomes:
\begin{equation}  \label{Eq:simplerProblem2}
	\Omega^* = \arg\min_\Omega \sum_{i=1}^{M}{\theta_{i} d_i^2 }
\end{equation}

This is simply a weighted sum of the square distances of the vertices from $\Omega$, with given constant weights that measure the exterior angles of the convex polygon.
Noting that the M exterior angles sum to $2\pi$, the optimizer of \eqref{Eq:simplerProblem2} is explicitly given by:
\begin{equation}  \label{Eq:omegaOptFinal}
	\Omega^*(x,y)=\left(\sum_{i=1}^{M}{\frac{\theta_{i}}{2\pi}x_i},\sum_{i=1}^{M}{\frac{\theta_{i}}{2\pi}y_i}\right)
\end{equation}

We see that a relatively straightforward calculation provides the optimal location $\Omega^*$ as a weighted average of the points $\bar{P}_1,\bar{P}_2,...,\bar{P}_M$, the weights being proportional to the turn angles at $\bar{P}_i$.

However this is exactly the Steiner center point of the convex polygon defined by the points $P_1,P_2,...,P_N$, as given by~\eqref{Eq:SteinerTraditional}.
\newline

Therefore, if we consider points inside the convex hull of $\{P_1,P_2,...,P_N\}$, the minimal coverage is attained by $\Omega^*$ at the Steiner center of the convex hull.
Suppose that the optimal coverage would be achieved at a point $\Omega^{ext}$ located outside the convex hull. 
Consider the set of points $\{P_1,P_2,...,P_N,\Omega^{ext}\}$. For this set, the optimal coverage considering points inside its convex hull will be achieved by some point $\Omega^{**}$- the Steiner center of that convex hull. This point will necessarily be a convex combination of $\{P_1,P_2,...,P_N,\Omega^{ext}\}$, where the weight of $\Omega^{ext}$ must be strictly positive (since $\Omega^{ext}$ is, by assumption, outside $CH\{P_1,P_2,...,P_N\}$).
Therefore the area of $\left(\bigcup_{i=1}^{N}{S_{P_i}(\Omega^{**})}\right)\cup S_{\Omega^{ext}}(\Omega^{**})$ will be smaller than $\left(\bigcup_{i=1}^{N}{S_{P_i}(\Omega^{ext})}\right)$ by the "optimality" of $\Omega^{**}$.
Hence 
\begin{equation}
\resizebox{.9\hsize}{!}{$
V\left( \bigcup_{i=1}^{N}{S_{P_i}(\Omega^{**})} \right) 
\leq V\left( \left(\bigcup_{i=1}^{N}{S_{P_i}(\Omega^{**})}\right)\cup S_{\Omega^{ext}}(\Omega^{**})\right) < V\left( \bigcup_{i=1}^{N}{S_{P_i}(\Omega^{ext})} \right)
$}
\end{equation}
where $V()$ denotes area, 
contradicting the assumption of optimality of $\Omega^{ext}$.

Hence the optimal location lies within the convex hull of the data points and is its Steiner center.
\end{proof}

\section{Concluding remarks} \label{Section:Conclusions}
This paper presented a novel characterization of the Steiner center as the point that provides the tightest disc coverage for the convex hull of the set of points in the plane.

We first showed that the convex hull of $N$ points in $\mathbb{R}^d$ is covered by the union of $d$-dimensional discs formed such that their diameters are the segments connecting some point $\Omega$ with each of the $N$ vertices.

Next we proved that in $\mathbb{R}^2$ the optimal location of $\Omega$ in the sense of minimizing the area difference between the union of discs and the convex hull is the well known Steiner center.
This interesting property is a nice addition to existing characterizations of the Steiner center.

For higher dimensions, we conjecture that the optimal point is the $\mathbb{R}^d$ Steiner center, but a proof is yet to be found.
Some numerical simulations that were performed in 3D seem to confirm this conjecture.

\subsection*{Acknowledgements}
\small{
We thank Danny Keren for discussion on the coverage problem, J\'anos Pach for pointing out the work of Elekes to us, and Sinai Robins for teaching us about the Steiner center.}

\bibliographystyle{plain}
\bibliography{steiner_bibfile}

\end{document}